\documentclass[conference,compsoc]{IEEEtran}
\ifCLASSOPTIONcompsoc
  \usepackage[nocompress]{cite}
\else
  \usepackage{cite}
\fi
\ifCLASSINFOpdf
\else
\fi
\usepackage{amsthm}
\usepackage{mathbbold}
\usepackage{latexsym}
\usepackage{amsmath}
\usepackage{amssymb}
\usepackage{txfonts}
\usepackage{xcolor,paralist,times}
\input{psfig.sty}
\usepackage{epstopdf}

\usepackage[final,
        bookmarks,
             bookmarksopen,
             colorlinks,
             final,
             linkcolor=blue,
             citecolor=brown,
             pdfstartview=FitH]%
{hyperref}

\usepackage{graphicx}

\newtheorem{Thm}{Theorem}
\newtheorem{Lem}[Thm]{Lemma}
\newtheorem{Def}{Definition}

\newtheorem{Cor}[Thm]{Corollary}
\newtheorem{Fact}{Fact}

\usepackage{algorithm}
\usepackage{algorithmic}

\newif\ifdraft\drafttrue

\ifdraft

\newcommand{\zd}[1]{{\color{green}{[{#1}---Duan]}}}

\newcommand{\lo}[1]{{\color{red}{[{#1}---Ong]}}}

\newcommand{\tc}[1]{{\color{blue}{[{#1}---Tian]}}}
\else

\newcommand{\zd}[1]{}

\newcommand{\lo}[1]{}

\newcommand{\tc}[1]{}
\fi
\newcommand\anglebra[1]{\langle #1 \rangle}

\begin{document}
%
\title{B\"{u}chi Determinization Made Tighter}


\author{\IEEEauthorblockN{Cong Tian and
Zhenhua Duan
}
\IEEEauthorblockA{ICTT and ISN Lab, Xidian University, Xi'an, 710071, China\\ Emails:\{ctian,zhhduan\}@mail.xidian.edu.cn }
}

\maketitle

\begin{abstract}
By separating the principal acceptance mechanism from the concrete
acceptance condition of a given B\"{u}chi automaton with $n$ states,
Schewe presented the construction of an equivalent deterministic
Rabin transition automaton with $o((1.65n)^n)$ states via \emph{history trees}, which can be
simply translated to a standard Rabin automaton with $o((2.26n)^n)$
states. Apart from the inherent simplicity, Schewe's construction
improved Safra's construction (which requires $12^nn^{2n}$ states).
However, the price that is paid is the use of $2^{n-1}$ Rabin pairs
(instead of $n$ in Safra's construction). Further, by introducing the \emph{later introduction record} as a
record tailored for ordered trees, deterministic automata with
Parity acceptance condition is constructed which
exactly resembles Piterman's determinization with Parity acceptance
condition where the state complexity is $O((n!)^2)$ and the index complexity is $2n$. 
 In this paper, we improve Schewe's construction to $2^{\lceil (n-1)/2\rceil}$  Rabin pairs with
the same state complexity. Meanwhile, we give a new determinization construction of Parity automata with the state complexity being $o(n^2(0.69n\sqrt{n})^n)$ and index complexity being $n$.

%
 
\end{abstract}


%
\IEEEpeerreviewmaketitle

\section{Introduction}\label{Sec:Intro}
The theory of automata over infinite words underpins much of formal
verification. In automata-based model checking
\cite{CGP00,Katoen99}, to decide whether a given system described by
an automaton satisfies a desired property specified by a B\"{u}chi
automaton \cite{Buchi62}, one constructs the intersection of the
system automaton with the complementation of the property automaton,
and checks its emptiness \cite{CGP00,Katoen99}. To complement
B\"{u}chi automata, Rabin, Muller, Parity as well as Streett
automata \cite{Thomas97} are often involved, since deterministic
B\"{u}chi automata are not closed under complementation. Recently,
co-B\"{u}chi automata have attracted much attention {because of} its
simplicity and its surprising utility \cite{BK09,BK10,BK11}.
Further, in LTL model checking, the transformation from LTL formulas
to B\"{u}chi automata is a key procedure, where generalized
B\"{u}chi automata are often used as an intermediary
\cite{CGP00,Katoen99}.

B\"{u}chi automata were first introduced as a tool for proving the
decidability of monadic second order logic of one successor (S1S)
\cite{Buchi62}. They are nearly the same as finite state automata
{except} for the acceptance condition:  a run of a finite state
automaton is accepting if a final state is visited at the end of the
run, whereas a run of a B\"{u}chi automaton is accepted if a final
state is visited infinitely often. This small difference enables
B\"{u}chi automata to accept infinite sequences instead of finite
ones. However, the close relationship between finite state and
B\"{u}chi automata does not mean that automata manipulations of
B\"{u}chi automata are as simple as those for finite state automata
\cite{RS59}. In particular, B\"{u}chi automata are not closed under
determinization. For a non-deterministic B\"{u}chi automaton, there
might not exist a deterministic B\"{u}chi automaton that accepts the
same language as the non-deterministic one does. That is,
deterministic B\"{u}chi automata are strictly less expressive than
the nondeterministic ones; while for a finite state automaton, a
simple subset construction is sufficient for efficient
determinization \cite{RS59}. Determinization of B\"{u}chi automata
requires automata with more complicated acceptance mechanisms, such
as automata with Muller's subset condition \cite{Muller63,McNau66},
Rabin and Streett's accepting pair conditions \cite{Safra88,MS95},
or Parity acceptance condition \cite{Piterman07,Sven09}, and so
forth.

Besides the close relation with complementation, determinization is
also useful in solving games and synthesizing strategies. The first
determinization construction for B\"{u}chi automata was introduced
by McNaughton by converting nondeterministic B\"{u}chi automata to
deterministic Muller automata with a doubly exponential blow-up
\cite{MS95}. Safra achieved an asymptotically optimal
determinization algorithm using \emph{Safra trees} \cite{Safra88}.
His method extends the powerset construction by branching out a new
computation path each time the given automaton reaches a final
state. Thus the states of the resulting automaton are not a set of
states but a set of tree structures. Safra's construction transforms
a nondeterministic B\"{u}chi automaton with $n$ states into a
deterministic Rabin automaton with $12^nn^{2n}$ states and $n$ Rabin
pairs. Piterman \cite{Piterman07} presented a tighter construction
by utilizing \emph{compact Safra trees} which are obtained by using
a dynamic naming technique throughout the construction of Safra
trees. With compact Safra trees, a nondeterministic B\"{u}chi
automaton can be transformed into an equivalent deterministic parity
automaton with $2n^nn!$ states and $2n$ priorities (can be
equivalently transformed to a deterministic Rabin automaton with the
same complexity and $n$ priorities). The advantage of Piterman's
determinization is to output deterministic Parity automata which is
easier to manipulate. Piterman pointed out in \cite{Piterman07} that it is an interesting question whether the $2n$ index priorities are indeed necessary.

By separating the principal acceptance mechanism from the concrete
acceptance condition of a B\"{u}chi automaton with $n$ states,
Schewe presented the construction of an equivalent deterministic
Rabin transition automaton with $o((1.65n)^n)$ states, which can be
simply translated to a standard Rabin automaton with $o((2.66n)^n)$
states \cite{Sven09}. Based on this construction, Schewe also obtained
an $O((n!)^2)$ transformation from B\"{u}chi automata to
deterministic parity automata that factually resembles Piterman's
construction (Liu and Wang \cite{LW09} independently present a
similar state complexity result to Piterman's determinization).
Schewe's construction is mainly based on a new data structure,
namely, \emph{history tree} which is an ordered tree with labels. Compared
to Safra's construction (whose complexity engendered a line of
expository work), Schewe's construction is simple and intuitive.
Subsequently, a lower bound for the transformation from
B\"{u}chi automata to deterministic Rabin transition automata is
proved to be $O((1.64n)^n)$ \cite{CZ09}. Therefore the state
complexity for Schewe's construction of Rabin transition automata is
optimal. For the ordinary Rabin acceptance condition, the
construction via \emph{history trees} also conducts the best upper-bound
complexity result $o((2.66n)^n)$. However, the price paid for that
is the use of $2^{n-1}$ Rabin pairs (instead of $n$ in Safra's
construction). 
Therefore, whether the index complexity in
Schewe's construction can be improved is an interesting problem.

In this paper, we reconstruct \emph{history trees} as \emph{history
trees with canonical (or spinal) identifiers} by adding an extra identifier on each node of the
tree. To reduce the index complexity and keep the state complexity
meanwhile, it is required that whenever a node $\tau$ occurs in a
\emph{history tree}, its identifer must keep the same with the one used in other
\emph{history trees} where $\tau$ occurs. We show that it is possible to
keep each $\tau$, in the \emph{history trees} throughout the
determinization construction, annotated by the same identifer (canonical identifier) ranging
over $2^{\lceil (n-1)/2\rceil}$ different identifiers. As a consequence, improved
constructions of deterministic Rabin transition automata and
ordinary Rabin automata are obtained which have the same state
complexity as Schewe's construction but reduces the index complexity
to  $2^{\lceil (n-1)/2\rceil}$. In addition, by utilizing the spinal (dynamic) identifiers, a new determinization construction of Parity automata with the state complexity being $o(n^2(0.69n\sqrt{n})^n)$ and index complexity being $n$ is obtained.

The rest of the paper is organized as follows. The next section
briefly introduces automata over infinite words. In Section
\ref{Sec:history}, Schewe's construction of deterministic Rabin
automata based on \emph{history trees} is presented. In section
\ref{sec:datastructure}, the main data structures \emph{history trees
with canonical or spinal identifiers} are presented. In the sequel, our improved
constructions of deterministic Rabin  automata and
Parity automata are presented
in Section \ref{Sec:compact} and \ref{sec:parity}, respectively. 

\section{Automata}\label{Sec:automata}
Let $\Sigma$ denote a finite set of symbols called an alphabet.
 An infinite word $\alpha$ is an infinite
sequence of symbols from $\Sigma$. $\Sigma^\omega$ is the set of all
infinite words over $\Sigma$.  We present $\alpha$ as a function
$\alpha:\mathbb{N}\rightarrow \Sigma$, where $\mathbb{N}$ 
is the set of non-negative integers. Thus, $\alpha(i)$ denotes the
letter appearing at the $i^{th}$ position of the word. In general,
$\mathsf{Inf}(\alpha)$ denotes the set of symbols from $\Sigma$
which occur infinitely often in $\alpha$. Formally,
$$\mathsf{Inf}(\alpha)=\{\sigma\in \Sigma\mid\exists^\omega n\in
\mathbb{N}:\alpha(n)=\sigma\}$$ Note that $\exists^\omega n \in
\mathbb{N}$ means there exist infinitely many $n$ in $\mathbb{N}$.

\begin{Def}[Automata] \rm
An automaton over $\Sigma$ is a tuple $A=(\Sigma, Q,\delta,
Q_0,\lambda)$, where $Q$ is a 
non-empty, finite set of states, $Q_0\subseteq Q$ is a set of initial
states, ${\delta} \subseteq
Q\times \Sigma\times Q$ 
is a transition relation, and $\lambda$ an acceptance condition.
\hfill{$\square$}
\end{Def}

A run $\rho$ of an automaton $A$ on an infinite word $\alpha$ is an
infinite sequence $\rho:\mathbb{N}\rightarrow Q$ such that
$\rho(0)\in Q_0$ and for all $i\in \mathbb{N}$,
$(\rho(i),\alpha(i),\rho(i+1))\in {\delta}$.  $A$ is said
deterministic if $I$ is a singleton, and for any $(q,\sigma,q')\in$
$\delta$, there exists no $(q,\sigma,q'')\in \delta$ such that
$q''=q'$, and non-deterministic otherwise. Similar to infinite
words, $\mathsf{Inf}(\rho)$ denotes the set of states from $Q$ which
occur infinitely often in $\rho$. Formally,
$$\mathsf{Inf}(\rho)=\{q\mid\exists^\omega n\in
\mathbb{N}:\rho(n)=q\}$$

Several acceptance conditions are studied in literature. We
present three of them here:
\begin{itemize}
\item  B\"{u}chi, where $\lambda\subseteq Q$, and
$\rho$ is accepted iff $\mathsf{Inf}(\rho)\cap
\lambda\not=\emptyset$.

%

\item Parity, where $\lambda=\{\lambda_1, \lambda_2, \ldots,
\lambda_{2k}\}$ with $\lambda_1\subset \lambda_2 \subset
\ldots\subset\lambda_{2k}=Q$, and $\rho$ is accepted if the minimal
index $i$ for which $\mathsf{Inf}(\rho)\cap \lambda_i\not=\emptyset$
is even.

\item Rabin, where $\lambda=\{(\lambda_1,\beta_1), (\lambda_2, \beta_2), \ldots,
(\lambda_{k},\beta_k)\}$ with $\lambda_i$, $\beta_i\subseteq Q$ and
$\rho$ is accepted iff for some $1\leq i\leq k$, we have that
$\mathsf{Inf}(\rho)\cap \lambda_i\not=\emptyset$ and
$\mathsf{Inf}(\rho)\cap \beta_i=\emptyset$.

%
%
\end{itemize}

An automaton accepts a word if it has an accepted run on it. The
accepted language of an automaton $A$, denoted by $L(A)$, is the set
of words that $A$ accepts.

We denote the different types of automata by three letter acronyms
in $\{D,N\}\times \{B,P,R\}\times \{W\}$. The first letter stands
for the branching mode of the automaton (deterministic or
nondeterministic); the second letter stands for the acceptance
condition type (B\"{u}chi, parity, or Rabin); and the third letter
indicates that the automaton runs on words. While acceptance
condition of an ordinary automaton is defined on states, the
acceptance condition of a transition automaton is defined on
transitions of the automaton. Accordingly, with respect to each type
of ordinary automata, we also have its transition version.


\section{Determinization via History Trees}\label{Sec:history}
This section presents Schewe's determinization via \emph{history trees}
\cite{Sven09}.

\subsection{History Trees}
A \emph{history tree} is a simplification of a \emph{Safra tree} \cite{Safra88} by
the omission of explicit names for nodes. It is factually an \emph{ordered
tree} with labels.

Let $\mathbb{N}_{{>}0}$ be the set of positive integers. An
\emph{ordered tree} $T \subseteq \mathbb{N}_{{>}0}^*$ is a finite,
prefix-closed and order-closed (with respect to siblings) subset of
finite sequences of positive integers. That is, if a sequence
$\tau=t_0t_1\ldots t_n$ is in $T$, then all proper prefixes of
$\tau$, $\tau'=t_0t_1\ldots t_m $, $m<n$, concatenated with $s$ where $s =
\epsilon$ or $s \in \{1, \cdots, t_{m+1}-1\}$ are also in $T$.  For a
node $\tau$ of an ordered tree $T$, we call the number of children
of $\tau$ its \emph{degree}, denoted by $deg_T(\tau)=|\{i\in
\mathbb{N}_{{>}0}\mid \tau\cdot i \in T\}|$.

Notice that a tree is not order-closed if, and only if, there exists
at least one node $x_1\ldots x_{n-1}x_n$ in the tree such that $x_n\geq 2$ and $x_1\ldots x_{n-1}(x_n-1)$ is not
in the tree. We call such a node an \emph{imbalanced node}. Further,
the greatest order-closed subset of a tree is called the set of
\emph{stable nodes}; and the other nodes are called \emph{unstable
nodes}. Intuitively, if $x_1\ldots
x_{n-1}x_n$ where $n\geq 1$ is an imbalanced node, 
all nodes $x_1\ldots x_{n-1}x$ where $x \geq x_n$, and their descendants, are unstable nodes.

\begin{Def} [History Trees \cite{Sven09}]\rm A history tree  for a given
nondeterministic B\"{u}chi automaton $A=(\Sigma,Q,Q_0,\delta,$ $\lambda)$
is a labelled tree $\anglebra{T,l}$, where $T$ is an ordered tree,
and $l:T\rightarrow 2^Q\setminus \{\emptyset\}$ is a labeling
function that maps the nodes of $T$ to non-empty subsets of $Q$,
such that the following three properties hold:
\begin{itemize}
\item P1: the label of each node is non-empty,
\item P2: the labels of different children of a node are
disjoint, and
\item P3: the label of each node is a proper superset of
the union of the labels of its children. \hfill{$\square$}
\end{itemize}
\end{Def}

Notice that in a \emph{history tree}, all the nodes are not explicitly
named, but an implicit name {(derived from the underlying \emph{ordered
tree})} is used. Fig.~\ref{fig:ht1}
shows an example of \emph{history tree} where the implicit names of nodes
are written in red. $a$, $b$, $c$, $d$, $e$, $f$, and $g$, are
states of the relative NBW. Note that in the rest part of the paper,
we often omit implicit names of nodes for simplicity.
\begin{figure}[htp]
\centerline{\psfig{figure=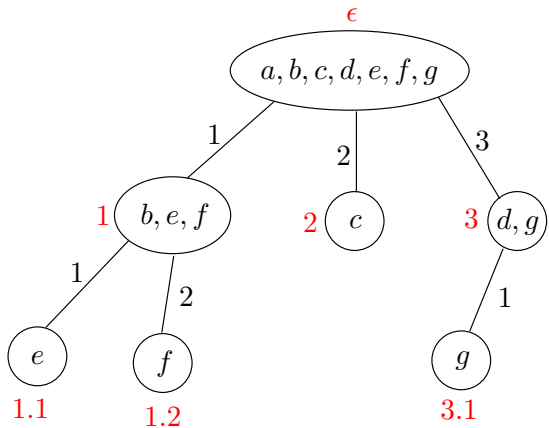,height=5.6cm}}\caption{A history
tree} \label{fig:ht1}
\end{figure}

\begin{Fact} \label{fact:nodes}\rm The number of the nodes in each \emph{history tree} of a
non-deterministic B\"{u}chi automaton with $n$ states is no more
than $n$.
\end{Fact}
\begin{proof}
By the definition of \emph{history trees}, the labels of siblings are
disjoint (P2) and the label of each node contains at leat a state
not appearing in any labels of its children (P3). That is any node
$v$ in the \emph{history tree} must have at least one state (in the
B\"{u}chi automaton) which is specific to $v$ other than any nodes
in the \emph{history tree}. Consequently, the \emph{history tree} can contain at
most $n$ different nodes.
\end{proof}

\paragraph*{\bf Minor modification of history trees} \emph{History
trees} presented in this paper are different from the ones given in
\cite{Sven09} where each node in a \emph{history tree} is implicitly named
as a sequence of \emph{non-negative} integers.
This small modification will not change the results presented in
\cite{Sven09} but useful in reducing  index complexity.

\subsection{Construction of History Trees}
For a nondeterministic B\"{u}chi automaton
$A=(\Sigma,Q,Q_0,$ $\delta,\lambda)$, the initial \emph{history tree} is
$\anglebra{T_0,l_0}=\anglebra{\{\epsilon\}, l_0(\epsilon)=Q_0}$ that
contains only one node $\epsilon$ that is labeled with the set of
initial states $Q_0$. For a \emph{history tree} $\anglebra{T,l}$, and an
input letter $\sigma\in \Sigma$ of $A$, a new \emph{history tree}, the
$\sigma$-successor $\anglebra{\widehat{T},\widehat{l}}$ of
$\anglebra{T,l}$, can be constructed in four steps:

\paragraph*{\scshape Step 1} We first construct a labelled tree $\anglebra{T',l':T'\rightarrow 2^Q}$ such
that
\begin{itemize}
\item $T' = T \cup \{ \tau' \mid
\tau'=\tau\cdot (deg_T(\tau)+1) \mbox{ for any } \tau \in T \}$.
That is, $T'$ is formed by taken all nodes from $T$, and new nodes
which are all node $\tau \in T$ appending with $deg_T(\tau)+1$.

\item the label $l'(\tau')$ of an old node $\tau'\in
T$ is the set $\delta(l(\tau'),\sigma)=\bigcup_{q\in
l(\tau')}\delta(q,\sigma)$ of $\sigma$-successors of the states in
the label of $\tau'$, and

\item the label $l'(\tau\cdot(deg_T(\tau)+1))$ of a new node $\tau\cdot (deg_T(\tau)+1)$ is the set of
final states in $\sigma$-successors of the states in the label of $\tau$, i.e.~$\delta(l(\tau),\sigma)\cap
\lambda$.
\end{itemize}

\paragraph*{\scshape Step 2} In this step, {\scshape P2} is re-established. We construct the
tree $\anglebra{T',l'': T'\rightarrow 2^Q}$, where $l''$ is inferred
from $l'$ by removing all states in the label of a node
$\tau'=\tau\cdot i$ and all its descendants if it appears in the
label $l'(\tau\cdot j)$ of an older sibling $(j<i)$.

\paragraph*{\scshape Step 3} {\scshape P1} and {\scshape P3} are re-established in
the third step. We construct the tree $\anglebra{T'', l'':
T''\rightarrow 2^Q}$ by (a) removing all descendants of nodes $\tau$
such that the label of $\tau$ is not a proper superset of the union
of the labels of the children of $\tau$; and (b) removing all nodes
$\tau$ with an empty label $l''(\tau)=\emptyset$. The stable nodes
whose descendant have been deleted due to rule (a) are
\emph{accepting} nodes.

\paragraph*{\scshape Step 4} The tree resulting from step 3 satisfies {\scshape P1} - {\scshape P3},
but it might be not order-closed, i.e. there exist unstable nodes in
the tree. We construct the $\sigma$-successor
$\anglebra{\widehat{T},\widehat{l}: T\rightarrow 2^Q\setminus
\{\emptyset\}}$ of $\anglebra{T,l}$ by ``compressing" $T''$ to an
order-closed tree, using the compression function $comp:T''
\rightarrow \omega^*$ that maps the empty word $\epsilon$ to
$\epsilon$, and $\tau\cdot i$ to $comp(\tau)\cdot j$, where
$j=|\{k<i\mid \tau\cdot k\in T''\}|+1$ is the number of older
siblings of $\tau\cdot i$ plus $1$\footnote{In \cite{Sven09},
$j=|\{k<i\mid \tau\cdot k\in T''\}|$ is defined to be the number of
older siblings of $\tau\cdot i$ since $0$ may occur in the name of
nodes.}. Note that the nodes that renamed during this step are
exactly those which are \emph{unstable}.

\subsection{Acceptance on Transitions}
Based on the above constructing procedure, for a given NBW, an
equivalent \emph{Deterministic Rabin Transition Automaton} (DRTW)
can be inductively established by:

\begin{itemize}
\item[(1)] Build the initial \emph{history tree};
\item[(2)] for each \emph{history tree}, by the four steps presented before, construct its $\sigma$-successor
history $T_{i+1}$ for each $\sigma\in \Sigma$;
\item[(3)] performed (2) repeatedly until  no new \emph{history trees} can be
created.
\end{itemize}
Now upon this underlying automata structure, deterministic Rabin
transition automata are defined.

\begin{Def} \rm The deterministic Rabin transition automaton that equivalents to a given NBW $A$ is $RT=(\Sigma,Q_{RT},Q_{RT0},\delta_{RT},\lambda_{RT})$, where alphabet $\Sigma$ is
the same as the one of $A$; $Q_{RT}$ is the set of \emph{history trees}
w.r.t $A$; $Q_{RT0}$ is the initial \emph{history tree}; $\delta_{RT}$ is the history
transition relation; and $\lambda_{RT}=((A_{\tau_1}, R_{\tau_1}), \ldots,
(A_{\tau_k}, R_{\tau_k}))$, $k\geq 1$, is the acceptance condition.
\hfill{$\square$}
\end{Def}

In each Rabin pair $(A_\tau, R_\tau)$, $\tau$ ranges over implicit names of nodes (in
a \emph{history tree}). $A_\tau$ is the set of transitions where $\tau$ is
accepting, and $R_\tau$ the set of  transitions in which $\tau$ is
unstable. For an input word $\alpha:\omega\rightarrow \Sigma$ we
call the sequence $\Pi=\anglebra{T_0,l_0}, \anglebra{T_1,l_1},
\ldots $ of \emph{history trees} that start with the initial \emph{history tree}
$\anglebra{T_0,l_0}$ and where, for every $i\in \omega$,
$\anglebra{T_i,l_i}$ is followed by $\alpha(i)$-successor
$\anglebra{T_{i+1},l_{i+1}}$, the history trace of $\alpha$. Let
$\Pi$ be the history trace of a word $\alpha$ on the deterministic
Rabin transition automaton. $\alpha$ is accepted by the automaton
if, and only if, $\mathsf{Inf}(\Pi)\cap A_{\tau_i}\not=\emptyset$
and $\mathsf{Inf}(\Pi)\cap R_{\tau_i}=\emptyset$ for some
$(A_{\tau_i},R_{\tau_i}) \in \lambda_{RT}$, $1\leq i\leq k$.

With these notations, the following useful results are proved in
\cite{Sven09}.
\begin{Lem}\label{sven1}
An $\omega$-word $\alpha$ is accepted by a nondeterministic
B\"{u}chi automaton $A$ if, and only if, there is a node $\tau\in
\mathbb{N}_{>0}^*$ such that $\tau$ is eventually always stable and
always eventually accepting in the history trace of $\alpha$.
\hfill{$\square$}
\end{Lem}

Let $RT$ be the deterministic Rabin transition automaton obtained
from the given non-deterministic B\"{u}chi automaton $A$. The
following corollary can easily be derived from Lemma \ref{sven1}.
\begin{Cor}
$L(RT)=L(A)$. \hfill{$\square$}
\end{Cor}

Eventually, the following main result for the determinization of
NBWs in Rabin transition acceptance condition is claimed
\cite{Sven09}.

%

\begin{Thm}\label{sven2}
For a given nondeterministic B\"{u}chi automaton with $n$ states, we
can construct a deterministic Rabin transition automaton with
$o((1.65n)^n)$ states and $2^{n-1}$ accepting pairs that recognizes
the language of $A$.    \hfill{$\square$}
\end{Thm}

The state complexity $o((1.65n)^n)$ is due to the result:
$$hist(n)\subset o((1.65n)^n)$$
where $hist(n)$ is the number of \emph{history trees} for a B\"{u}chi
automaton with $n$ states; the index complexity is based on the
number of (identifiable) nodes in the \emph{history trees} \cite{Sven09}.
Note that even though there are at most $n$ nodes in each \emph{history
tree},  there are a total of $2^{n-1}$ different nodes (identified by
name) which may be involved in the history transitions of the
B\"{u}chi automaton with $n$ states\footnote{The essential principle will be formally analyzed in Fact \ref{fact:full}.}.

Due to the lower-bound result in \cite{CZ09}, the state complexity
for the construction of deterministic Rabin transition automaton is
tight. However, whether the exponential index complexity can be
improved is interesting.

%
%
%

In \cite{Sven09}, by enriching \emph{history trees} with information about
which node of the resulting tree was accepting or unstable in the
third step of the transition, ordinary deterministic Rabin automata
can also be achieved.

\begin{Thm}\label{Thm:sven}
For a given nondeterministic B\"{u}chi automaton with $n$ states, we
can construct a deterministic Rabin automaton with $o((2.66n)^n)$
states and $2^{n-1}$ accepting pairs that recognizes the language
$L(A)$ of $A$.    \hfill{$\square$}
\end{Thm}

Further, by introducing the \emph{later introduction record} as a
record tailored for ordered trees, deterministic automata with
Parity acceptance condition is constructed \cite{Sven09} which
exactly resembles Piterman's determinization with Parity acceptance
condition \cite{Piterman07}.

\begin{Thm}
For a given nondeterministic B\"{u}chi automaton with $n$ states, we
can construct a deterministic parity automaton with $O(n!^2)$ states
and $2n$ priorities that recognizes the language $L(A)$ of $A$.
\hfill{$\square$}
\end{Thm}

Note that the original state complexity of Piterman's construction
is $2n^n n!$. By giving a better analysis, Liu and Wang,
independently, achieved a similar result ($2n(n!)^2$ states and $2n$
priorities) for Piterman's construction \cite{LW09}.

Since Parity acceptance condition is more general than
Rabin acceptance condition, we also have the transformation from NBW
to DRW with $O(n!^2)$ states and $n$ priorities. However, in state
of the art for the upper bound complexity of the transformation from
NBW to DRW, it is hard to say which one in ($o((2.66n)^n)$,
$2^{n-1}$) and ($O(n!^2)$, $n$) is better. Note that the first
element in the pair denotes the state complexity while the second
one in the pair indicates the index complexity. Further, Piterman pointed out in \cite{Piterman07} that it is an interesting question whether the $2n$ index priorities are indeed necessary in his construction.

\section{Ordered Trees with Identifiers}\label{sec:datastructure}

This section presents \emph{ordered trees with identifiers} that are
\emph{ordered trees} with each node annotated by a unique identifier. Specifically, two kinds of identifiers, namely canonical (statical) and spine (dynamic) identifiers, are proposed.

\subsection{Height of Nodes in Ordered Trees} \label{Sec:semi}
Given an ordered tree $T$, the height of a node $\tau=t_1t_2\ldots
t_n\in T$ is $\mathit{h}(\tau) :=
\sum\limits_{i=1}^{len(\tau)}t_i$. Specially, for the root node $\epsilon$, $\mathit{h}(\epsilon) :=
0$. For instance, the height of each node in the ordered tree in Fig. \ref{fig:sem} (1) is written inside the nodes as depicted in Fig. \ref{fig:sem}
(2).
\begin{figure}[htp]
\centerline{\psfig{figure=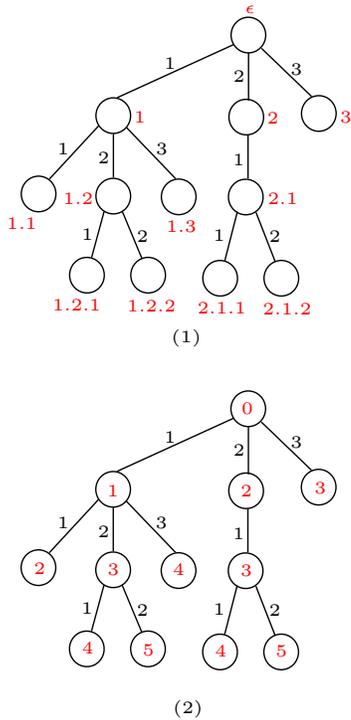,height=9.5cm}}\caption{Height of nodes in ordered trees} \label{fig:sem}
\end{figure}

%
%
%
%

With respect to
height of nodes in \emph{ordered trees}, order closedness for \emph{ordered trees} is defined
below.

\begin{Def}[Order Closedness] \rm An \emph{ordered tree} is \emph{order-closed} if, and only if, it satisfies the following three conditions:
\begin{itemize}
\item {\scshape C1:} The height of the root node is $0$;
\item {\scshape C2:} for each node with height being $n$,  the height of its left sibling (older) is $n-1$ if it
exists; and
\item {\scshape C3:} for each node with height being $n$, the height of its leftmost child
is $n+1$ if it exists.  \hfill{$\square$}
\end{itemize}
\end{Def}

If a node violates at least one of the above three conditions, it is an \emph{imbalanced node}. Further, we call an
imbalanced node, its younger siblings and their respective
descendants \emph{unstable nodes}; and the rest are
\emph{stable nodes}. (Thus, in particular, an imbalanced
node is unstable.) Accordingly, we observe that whether a node
in an \emph{ordered tree} is imbalanced (respectively unstable and
stable) depends on the height, which contains less information than implicit names, of nodes in the \emph{ordered trees}.
Note that if the \emph{ordered trees} we use are exactly the same as Schewe's \cite{Sven09},
i.e.~each integer in the name sequence is greater or equal to $0$,
the above properties for \emph{ordered tree} will not hold.

We define a partial order relation $\preceq$ over nodes in a
tree. Let $\tau=t_1t_2\cdot\cdot\cdot t_n$ and $\tau'$ be nodes of
an \emph{ordered tree} $T$. We define $\tau'\preceq \tau$ just if $\tau'$
is a proper prefix $t_0t_1\ldots t_m$, $m<n$, of $\tau$ concatenated with
$s$ where $s = \epsilon$ or $s \in \{1,\cdots, t_{m+1}-1\}$.
The partial
order relation $\preceq$ precisely characterizes the relationship
(in stable or unstable situation) among nodes in an \emph{ordered tree}.
Suppose $\tau'\preceq\tau$ for two different nodes in $T$. In case
$\tau'$ is unstable, $\tau$ is also unstable. Accordingly, w.r.t.
each node $\tau$ in a tree $T$, a maximal chain $C(\tau)=\{\tau'\mid
\tau'\in T,\mbox{ and }\tau'\preceq \tau\}$ can be obtained. It is pointed out that for each node $\tau$ of an \emph{ordered tree}, $h(\tau)=|C(\tau)|$.
For instance, for $\tau_5$ in Fig.~\ref{fig:partial},
\begin{figure}[htp]
\centerline{\psfig{figure=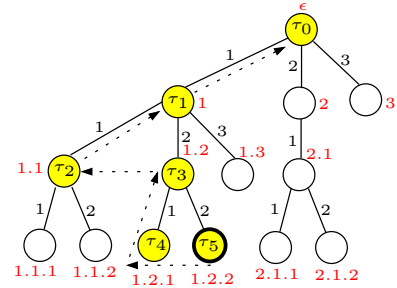,height=3.8cm}}\caption{Partial
order relation} \label{fig:partial}
\end{figure}
a maximal chain
$\tau_0\preceq\tau_1\preceq\tau_2\preceq\tau_3\preceq\tau_4\preceq\tau_5$
can be obtained. If $\tau_i$ is unstable for some $0\leq i\leq 4$,
so is $\tau_5$. But if $\tau_5$ is stable, $\tau_i$ is also stable  for all $0\leq i\leq 4$. Further, the number of the partial order
chains in a tree depends on the number of the leaf nodes which is the
right-most (youngest) child of its parent, e.g. nodes 1.1.2, 1.2.2,
1.3, 2.1.2, and 3 in Fig.~\ref{fig:partial}.

\subsection{$n$-Ordered Tree with Canonical Identifiers}
We first define \emph{$n$-full ordered tree} which contains all the nodes that
are possible to occur in an \emph{ordered tree} with $n$ nodes (called \emph{$n$-ordered tree} for convenience).

\begin{Def}[$n$-Full Ordered Tree] The \emph{$n$-full ordered tree}, $n\in \mathbb{N}_{{>}0}$, is an \emph{ordered
tree} where for each non-leaf node $\tau$ in the tree,

\begin{itemize}
\item the right-most child
node $\tau'$ of $\tau$ is a leaf node with $|C(\tau)|=n$;
\item other children nodes of $\tau$ are non-leaf nodes.
\hfill{$\square$}
\end{itemize}
\end{Def}

The black part in Fig. \ref{fig:fullordertree} (1) - (6) shows $2$
to \emph{$7$-full ordered tree}, respectively. Intuitively, \emph{$n$-full
ordered tree}, $n>1$, can be obtained from \emph{$(n-1)$-full ordered tree}
$T$ by adding a new right-most child node to each node in $T$.

\begin{figure}[htp]
\centerline{\psfig{figure=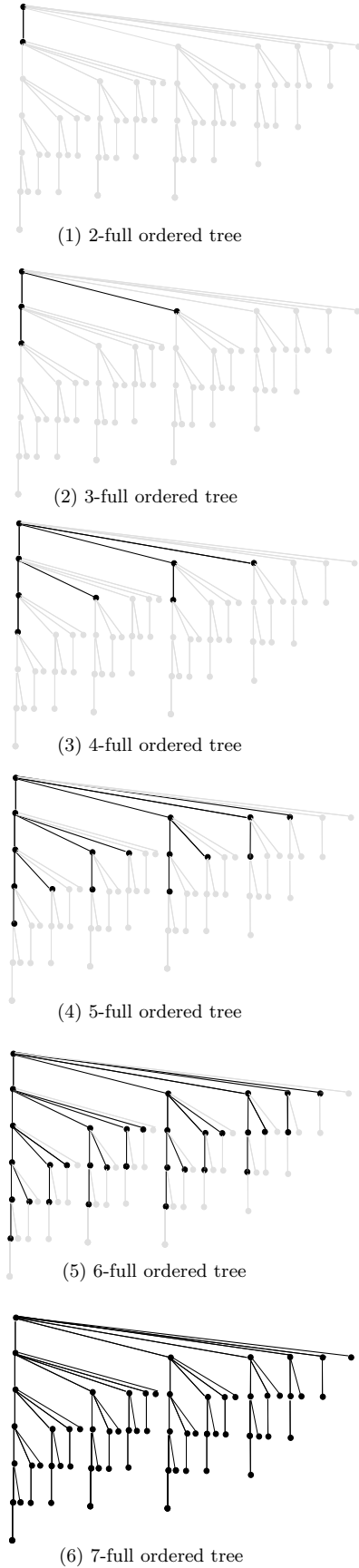,height=23cm}}\caption{$n$-Full
ordered tree} \label{fig:fullordertree}
\end{figure}

\begin{Fact}\label{fact:full} In the \emph{$n$-full ordered tree}, there are $2^{n-1}$ different nodes.
\end{Fact}

\begin{proof}
The fact is obvious in case $n=1$ or $2$. Suppose the fact holds
when $n=i$. That is there are totally $2^{i-1}$ nodes in \emph{$i$-full
ordered tree} $T$. Subsequently, \emph{$(i+1)$-full ordered} tree $T'$ can
be obtained by creating a new right-most child for each node in $T$.
Thus, there are  $2^i$ nodes in \emph{$(i+1)$-full ordered tree}. So the
fact holds.
\end{proof}

Fact \ref{fact:full} shows that there are totally $2^{n-1}$ different nodes (identified by implicit names) possible to occur in an \emph{$n$-ordered tree}.
Thus, the index complexity of Schewe's determinization is $2^{n-1}$ since implicit names of nodes in ordered trees containing $n$ nodes are utilized as the subscript when defining Rabin acceptance pairs.
Fig. \ref{fig:ordertree}
\begin{figure}[htp]
\centerline{\psfig{figure=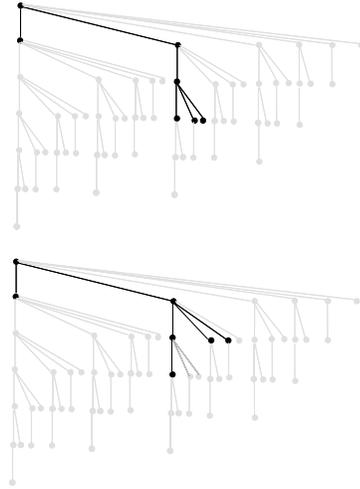,height=6.5cm}}\caption{
Ordered trees in full ordered tree} \label{fig:ordertree}
\end{figure}
shows two different \emph{$7$-ordered trees} (\emph{ordered trees with $7$ nodes}) in the \emph{$7$-full
ordered tree} depicted in Fig. \ref{fig:fullordertree} (6).

Let $\tau$ and $\tau'$ be two different nodes  in the \emph{n-full ordered
tree}. $\tau$ and $\tau'$ can occur in the same \emph{n-ordered tree} if, and only if, $|C(\tau)\cup C(\tau')|\leq n$. This indicates that if we assign a node $\tau$ of the \emph{$n$-full ordered tree} a flag $f$ such that the flag is unique in ordered trees with no more that $n$ nodes where $\tau$ may occur, any other node $\tau'$ where $|C(\tau)\cup C(\tau'')|\leq n$ should be assigned a flag different to $f$. To use less flags meanwhile, a node $\tau''$ where $C(\tau)\cup C(\tau'')>n$ can share the same flag with $\tau$. To further reduce the number of flags, we can use the height of each node as one property, and then for two nodes with different heights in an ordered tree, they can share the same flag. Accordingly, each node can be uniquely identified by its height in addition to flag in each ordered tree containing no more than $n$ nodes.

\begin{Lem}\label{Lem:fullnodes}
 $2^{\lceil (n-1)/2\rceil-1}$ different flags are enough in \emph{$n$-full ordered tree} such that each node in it can be uniquely identified in every \emph{$n$-ordered tree}.
\end{Lem}

\begin{proof}
Let $\tau$ be a node in the \emph{$n$-ordered tree} with $h(\tau)=e$.
 There are  $2^{e-1}$ different nodes with $h(\tau)=e$ in total that are possible to occur in an \emph{$n$-ordered tree} where $\tau$ appears in case  $e\leq\lceil (n-1)/2\rceil$; and $2^{n-e-1}$ otherwise.
Thus, in case $e=\lceil (n-1)/2\rceil$, the maximal number of flags are required. Therefore, at least $2^{\lceil (n-1)/2\rceil-1}$ different flags are required in \emph{$n$-full ordered tree} such that each node in any \emph{$n$-ordered tree} is uniquely identified.
\end{proof}

Now we define \emph{$n$-full ordered tree} with identifiers where
every node in each \emph{$n$-ordered trees} is
annotated by a unique identifier.

\begin{Def}[$n$-Full Ordered Tree with Identifiers] \emph{$n$-full ordered tree} with identifiers is a pair $\langle T, I\rangle$
where $T$ is \emph{$n$-full ordered tree}, and $I:T\rightarrow (\mathbb{N},\mathbb{N}_{{>}0} )$ maps
each node $\tau$ in $T$ to its identifier $(h,f)$. Here $h$ is the height, $h(\tau)$, of $\tau$,  and $f$ the flag of $\tau$, denoted by $flag(\tau)$,
such that for any two
different nodes, say $\tau$ and $\tau'$, in $T$,
$flag(\tau)\not=flag(\tau')$ if $|C(\tau)\cup C(\tau')|\leq n$ and $h(\tau)=h(\tau')$.
\hfill{$\square$}
\end{Def}

There are different ways (functions) to assign each node
in \emph{$n$-full ordered tree} a proper flag to obtain an \emph{$n$-full ordered tree with identifiers} such that exactly  $2^{\lceil (n-1)/2\rceil-1}$ different flags are utilized. Among them, we suppose function $flag_c:T\rightarrow \mathbb{N}_{{>}0}$, where $T$ is the $n$-full ordered tree, results in a canonical flag for each node in the $n$-full ordered tree. With this basis, \emph{$n$-Full Ordered Tree with Canonical Identifiers} is defined. 



\begin{Def}(\emph{$n$-Full Ordered Tree with Canonical Identifiers}) \emph{$n$-full ordered tree with canonical identifiers} is a pair $\langle T, I_c\rangle$
where $T$ is the \emph{$n$-full ordered tree}, and $I_c:T\rightarrow (\mathbb{N},\mathbb{N}_{{>}0} )$ maps
each node $\tau$ in $T$ to its canonical identifier $(h,f_c)$. Here $h$ is the height, $h(\tau)$, of $\tau$  and $f_c$ the canonical flag of $\tau$ obtained by $flag_c(\tau)$.
\hfill{$\square$}
\end{Def}

Fig. \ref{fig:flagc} depicts $7$-full ordered tree with canonical
identifers.
\begin{figure}[htp]
\centerline{\psfig{figure=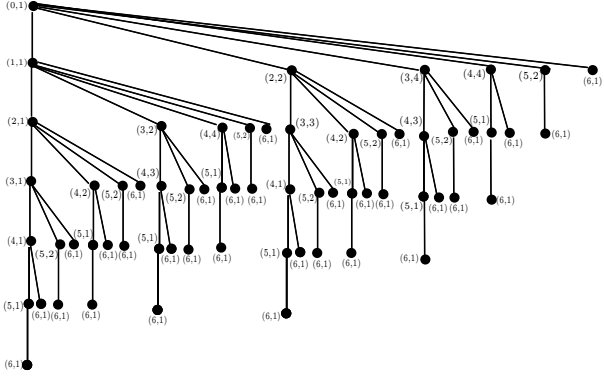,height=5.0cm}}\caption{$7$-full
ordered tree with canonical identifiers} \label{fig:flagc}
\end{figure}


\begin{Def}[$n$-ordered tree with canonical identifiers] Each ordered tree with $n$ nodes in an $n$-full ordered tree is called an $n$-ordered tree with canonical identifiers.
\hfill{$\square$}
\end{Def}
On \emph{$n$-full ordered tree with canonical identifiers},  totally $2^{n-2}$ different \emph{n-ordered trees with canonical identifiers} can be found. Let $T$ be one of the \emph{n-ordered trees with canonical identifiers}. Each node in $T$ can be uniquely distinguished by its identifier. Accordingly, to reduce the index complexity of the determinization via history trees, in stead of the implicit names of nodes in \emph{ordered trees}, we can identify each node in an \emph{ordered tree} by its canonical identifer. This is useful in Section \ref{Sec:compact} to improve the index complexity.

\subsection{$n$-Ordered Tree with Spinal Identifiers}

Let $rml=\tau_1\ldots\tau_x$ be a sequence of all  leaf nodes with no younger siblings in an $n$-ordered tree.
With respect to each $\tau_i$, $1\leq i\leq x$, in $rml$, a \emph{spine}
$sp_i$ is defined by $C(\tau_i)\setminus(C(\tau_1)\cup\ldots\cup
C(\tau_{i-1}))$. Note that here $x\leq \lfloor\frac{n}{2}\rfloor$ (see Fact \ref{fact:x}).  By assigning each spine a unique flag (integers in $\{1, \cdots, \lfloor\frac{n}{2}\rfloor\}$) with respect to the order it appearing in the sequence, each node in the $n$-ordered tree can be uniquely identified by considering both its height and spinal flag. Note that all nodes in a spine $sp$ form a finite chain under the partial order relation $\preceq$ with a leaf node without younger sibling being the top. The spinal flag of a node relies on the order the leaf nodes with no younger siblings appear in $rml$. That is whenever the spinal flags of all the leaf nodes with no younger siblings in a tree is given, the spinal flags of other nodes are obtained. Each leaf node with no younger siblings  corresponds to one unique spine.

\begin{Def}[$n$-ordered trees with spinal identifiers]
An \emph{$n$-ordered tree with spinal identifiers} is a pair $\langle T, I_s\rangle$
where $T$ is an \emph{$n$-ordered tree}, and $I_s:T\rightarrow (\mathbb{N},\mathbb{N}_{{>}0} )$ maps
each node $\tau$ in $T$ to its spinal identifier $(h,f_c)$. Here $h$ is the height, $h(\tau)$, of $\tau$,  and $f_s$ the spinal flag, $flag_s(\tau)$ of $\tau$.
\hfill{$\square$}
\end{Def}

\begin{Fact}\label{fact:x}
There are at most $\lfloor\frac{n}{2}\rfloor$ leaf nodes with no younger siblings in an $n$-ordered tree.
\end{Fact}
\begin{proof}
For the ordered trees with the numbers of nodes being $1$ and $2$, only $1$ leaf node with no younger siblings can be found.
Then it can be observed that anytime two new nodes join in the ordered tree, at most $1$ more leaf node with no younger siblings can be created. Thus, in ordered trees with $n$ nodes, at most $1+\lfloor \frac{n-2}{2}\rfloor=\lfloor \frac{n}{2}\rfloor$ leaf nodes with no younger siblings are contained.
\end{proof}

\begin{Lem}\label{Lem:spine}
For an $n$-ordered tree, at most $\lfloor\frac{n}{2}\rfloor!$ different $n$-ordered trees with spinal identifiers can be obtained.
\end{Lem}
\begin{proof}
For a given sequence of leaf nodes with no younger siblings in an $n$-ordered tree, we flag these nodes with its sequence number starting at $1$. As a result, by the definition of $n$-ordered trees with spinal identifers,
 a unique $n$-ordered tree with spinal identifiers is obtained.
 Thus, the number of $n$-ordered trees with spinal identifiers on a given $n$-ordered trees is upon to the number of sequences of leaf nodes with no younger siblings in the $n$-ordered tree. Therefore, by Fact \ref{fact:x}, at most $\lfloor\frac{n}{2}\rfloor!$ different $n$-ordered trees with spinal identifiers can be obtained.
\end{proof}


\section{Determinization of B\"{u}chi Automata  via History Trees with Canonical Identifiers} \label{Sec:compact}
This section presents the construction of deterministic Rabin
transition automata via \emph{history trees with canonical identifiers} from
nondeterministic B\"{u}chi automata.

\subsection{History Trees with Canonical Identifiers}
\emph{History trees with canonical Identifiers} are built upon \emph{ordered trees with
canonical identifiers}.

\begin{Def} [History Trees with Canonical Identifiers]\rm A \emph{history tree with canonical identifiers}
for a given nondeterministic B\"{u}chi automaton
$B=(\Sigma,Q,Q_0,\delta,\lambda)$ with $n$ states is a triple
$\anglebra{T,I_c,l}$ where $\langle T, I_c\rangle$ is an \emph{ordered tree
with canonical identifiers}, and $l:T\rightarrow 2^Q\setminus
\{\emptyset\}$ labels each node of $T$ with a non-empty subset of
$Q$, satisfying the following:
\begin{itemize}
\item[{\scshape P1}:] The label of each node is non-empty.
\item[{\scshape P2}:] The labels of different children of a node are
disjoint.
\item[{\scshape P3}:] The label of each node is a proper superset of
the union of the labels of its children. \hfill{$\square$}
\end{itemize}
\end{Def}

By Fact \ref{fact:nodes},
there are at most $n$ nodes in the \emph{history trees with canonical
identifiers} for a given B\"{u}chi automaton with $n$ states.
Fig.~\ref{fig:shta1}
\begin{figure}[htp]
\centerline{\psfig{figure=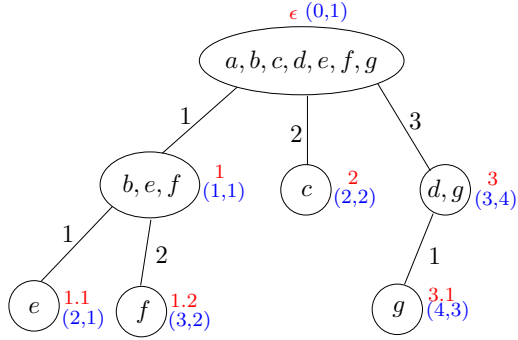,height=4.5cm}}\caption{A history
tree with canonical identifiers} \label{fig:shta1}
\end{figure}
shows a \emph{history tree with canonical identifiers} where the label of a node
is written inside the circle and the identifiers of nodes are in blue. It
is a modification of the history tree (of a B\"{u}chi automaton with $7$ states) in Fig.~\ref{fig:ht1} with
each node annotated by an extra identifier adopt from \emph{7-full ordered tree with canonical identifiers} in Fig. \ref{fig:flagc}.

\subsection{Construction of History Trees with Canonical Identifiers}


Fix a nondeterministic B\"{u}chi automaton $B=(\Sigma,Q,Q_0,$
$\delta,\lambda)$. The initial  \emph{history tree with canonical
identifiers} $\anglebra{T_0,I_{c,0},l_0}$ contains only one node
$\epsilon$ with $I_{c,0}(\epsilon)=(0,1)$  and label being the set
of initial states $Q_0$ of $B$.

Given a \emph{history tree with canonical identifiers} $\anglebra{T,I_c,l}$, and an
input letter $\sigma\in \Sigma$ of B\"{u}chi automaton $B$, the
$\sigma$-successor \emph{history tree with canonical identifiers}
$\anglebra{\widehat{T},\widehat{I_c},\widehat{l}}$ of
$\anglebra{T,I_c,l}$ is constructed in five steps below.

\paragraph*{\scshape Step 1} We construct the  labeled tree with canonical identifiers $\anglebra{T',I_c',l':T'\rightarrow 2^Q}$ such
that
\begin{itemize}
\item[(1)] $T' = T \cup \{ \tau'=\tau\cdot(deg_{T}(\tau)+1) \mid
 \mbox{ for any } \tau \in T \}$.
That is, $T'$ is formed by taken all nodes from $T$ and creating the
new child node $\tau\cdot(deg_{T}(\tau)+1)$ to any node $\tau \in
T$;

\item[(2)] the label $l'(\tau)=\delta(l(\tau),\sigma)$ of an old node $\tau\in
T$ is the set $\delta(l(\tau),\sigma)=\bigcup_{q\in
l(\tau)}\delta(q,\sigma)$ of $\sigma$-successors of the states in
the label of $\tau$;

\item[(3)] the label $l'(\tau')=\delta(l(\tau),\sigma)\cap
\lambda$ of a new node $\tau'=\tau\cdot (deg_{T}(\tau)+1)$ is the
set of final $\sigma$-successors of the states in the label of
$\tau$;

\item[(4)] the identifiers of nodes are adopted from the \emph{n-full ordered
tree with canonical identifiers}. Here $n=|Q|$.

%

\end{itemize}

To illustrate the five steps of construction, we borrow the examples
from \cite{Sven09} for ease of comparison with Schewe's method.
Fig.~\ref{fig:transition} shows all transitions for an input letter
\begin{figure}[htp]
\centerline{\psfig{figure=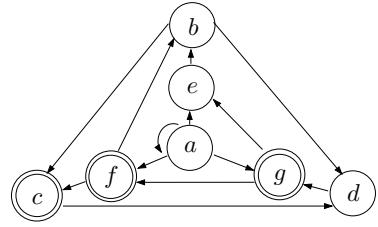,height=3.0cm}}\caption{Relevant
fragment of a B\"{u}chi automaton} \label{fig:transition}
\end{figure}
$\sigma$ from the states in the \emph{history tree with canonical identifiers} in
Fig.~\ref{fig:shta1}. The double circles indicate that the states
$c$, $f$, and $g$ are final states.

Fig.~\ref{fig:shta2}  depicts the tree resulting from the \emph{history
tree with canonical identifiers} in Fig. \ref{fig:shta1} for the B\"{u}chi
automaton and transition from Fig. \ref{fig:transition} after
{\scshape Step 1} of the construction procedure. Each node of the
tree from Fig.~\ref{fig:shta1} has spawned a new child (marked by
bold circle), whose label may be empty if, upon reading the input
letter from any state in the label of the parent node, none of the
new transition state is a final state.
The canonical identifers of nodes are adopted from \emph{7-full ordered tree with canonical identifers} in Fig. \ref{fig:flagc}.
The labels of nodes in gray will be deleted from the respective
label in {\scshape Step 2}. Note that in Fig. \ref{fig:shta2}, there may exist different nodes annotated with the same identifers, i.e. 1.3 and 4 are both annotated by $(4,4)$, since the tree resulted in this step may contain more than $n$ nodes (no more than $2n$ nodes in fact). However, at most one of the nodes with a same identifier will be kept in the eventually \emph{history tree with canonical identifers} resulted in {\scshape Step 5}.

\begin{figure}[htp]
\centerline{\psfig{figure=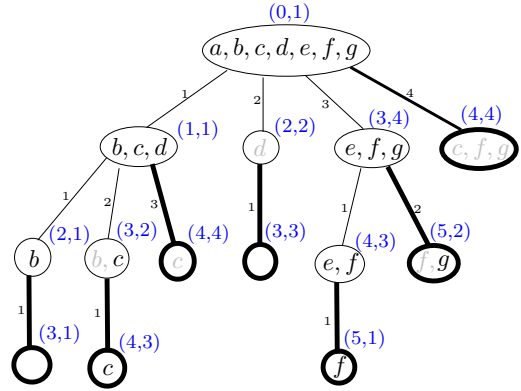,height=5.2cm}}\caption{{\scshape
Step 1} of the construction procedure} \label{fig:shta2}
\end{figure}

\paragraph*{\scshape Step 2} In this step, {\scshape P2} is re-established. We construct the
tree $\anglebra{T',I_c',l'': T'\rightarrow 2^Q}$, where $l''$ is
derived from $l'$ by removing each state {$q$} in the label of a
node $\tau$ and all its descendants if {$q$} appears in the label
$l'(\tau)$ of an older sibling.

Fig.~\ref{fig:shta3} shows the tree that results from {\scshape Step
2} of the construction procedure. In the resultant tree, the labels
of the siblings are pairwise disjoint, but may be empty, and the
union of the labels of the children of a node are not required to
form a proper subset of their parent's label. The part in dash lines
will be deleted in {\scshape Step 3}.

\begin{figure}[htp]
\centerline{\psfig{figure=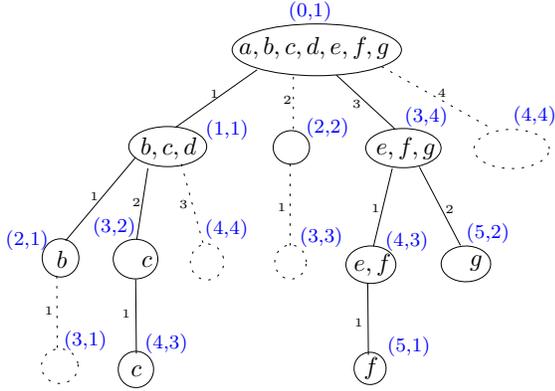,height=5.2cm}}\caption{{\scshape
Step 2} of the construction procedure} \label{fig:shta3}
\end{figure}

\paragraph*{\scshape Step 3} {\scshape P1} is re-established in
the third step. We construct the tree $\anglebra{T'', I_c'',l'':
T''\rightarrow 2^Q}$ by  removing all nodes $\tau$ with an empty
label $l''(\tau)=\emptyset$.

The resulting tree is depicted in Fig.~\ref{fig:shta4}.  The part in
dash lines will be deleted in the next step.
\begin{figure}[htp]
\centerline{\psfig{figure=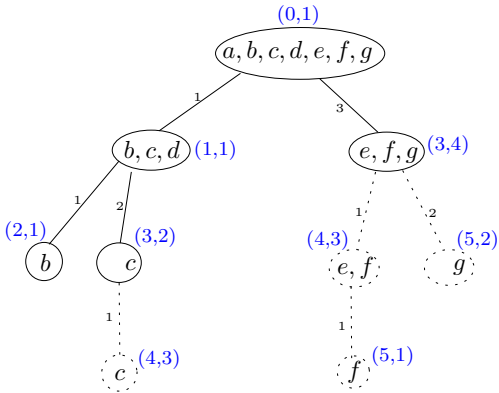,height=5.2cm}}\caption{{\scshape
Step 3} of the construction procedure} \label{fig:shta4}
\end{figure}

\paragraph*{\scshape Step 4} {\scshape P3} is re-established in
the this step. We construct the tree $\anglebra{T''', I_c''', l''':
T'''\rightarrow 2^Q}$ by removing all descendants of nodes $\tau$ if
the label of $\tau$ is not a proper superset of the union of the
labels of the children of $\tau$. The nodes whose descendant have
been deleted because of this rule are called \emph{accepting} nodes.
For each accepting node, we add a symbol $\oplus$ on it to
explicitly indicate that the node is accepting at the current step.

The resulting tree is depicted in Fig.~\ref{fig:shta5} (1).
\begin{figure}[htp]
\centerline{\psfig{figure=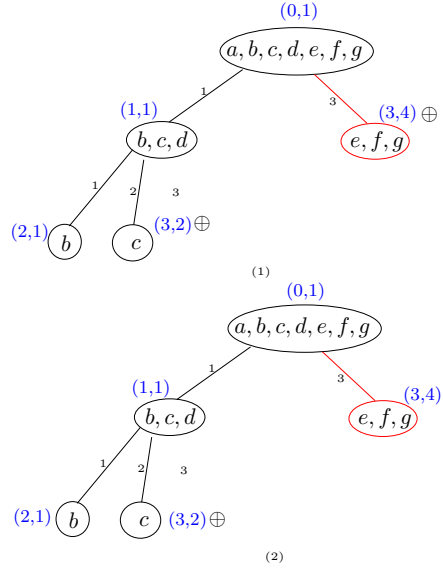,height=7.5cm}}\caption{{\scshape
Step 4} of the construction procedure} \label{fig:shta5}
\end{figure}
The labels of the siblings are pairwise disjoint, and form a proper
subset of their parent's label, but the tree might not be
order-closed. Two nodes with identifers being $(3,2)$ and $(3,4)$ are both accepting, as
marked by $\oplus$. The nodes that will be renamed for establishing
order closedness in {\scshape Step 5} are drawn in red.

\paragraph*{\scshape Step 5}
 The underling tree resulting from {\scshape Step 4} satisfies properties {\scshape P1-P3}, but it may not be order-closed, i.e.~there exist some imbalance
nodes in the tree (e.g. $(3,4)$). We construct $\sigma$-successor
$\anglebra{\widehat{T},\widehat{I_c},\widehat{l}: T\rightarrow
2^Q\setminus \{\emptyset\}}$ of $\anglebra{T,I_c,l}$
  by removing all the $\oplus$ symbols
and ``compressing" $T'''$ to an order-closed tree, using the
compression function $comp:T''' \rightarrow \omega^*$ that maps the
empty word $\epsilon$ to $\epsilon$, and $\tau\cdot i$ to
$comp(\tau)\cdot j$, where $j=|\{k<i\mid \tau\cdot k\in T'''\}|+1$
is the number of older siblings of $\tau\cdot i$ plus $1$.
Eventually, all the \emph{unstable} nodes, i.e.~imbalanced nodes as
well as the descendants of the imbalanced nodes, and the younger
siblings (and their descendants) of the imbalanced nodes, are
renamed, and their identifiers are also updated. Further, for each node
$\tau$ that is properly renamed in this step, we add a $\ominus$
symbol on the node before renaming (in the tree obtained in
{\scshape Step 4}) to mean that this node is unstable in {\scshape
Step 5} if there is no $\oplus$ symbol on the node, otherwise we
just remove the $\oplus$ symbol from the node.

Fig.~\ref{fig:shta6} shows the $\sigma$-successor
$\anglebra{\widehat{T},\widehat{I_c},\widehat{l}: T\rightarrow
2^Q\setminus \{\emptyset\}}$ of $\anglebra{T,I_c,l}$ obtained in
{\scshape Step 5} and Fig.~\ref{fig:shta1} (2) is the tree
originally obtained in {\scshape Step 4} but updated by {\scshape
Step 5}.

\begin{figure}[htp]
\centerline{\psfig{figure=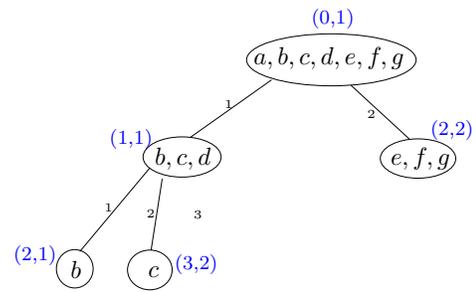,height=3.8cm}}\caption{{\scshape
Step 5} of the construction procedure} \label{fig:shta6}
\end{figure}

Notice that the symbols $\oplus$ and $\ominus$ occur only in the
trees created in the intermediate steps within the transitions, not
in the \emph{history trees with canonical identifiers} (states in the
deterministic Rabin automata) eventually constructed.

\subsection{Deterministic Rabin Transition Automata}
Based on the five-step construction procedure, for a given NBW, an
equivalent deterministic Rabin transition automaton can be
inductively constructed by:
\begin{itemize}
\item[(1)] Build the initial \emph{history tree with canonical
identifiers} $(T_0,I_{c,0},l_0)$;
\item[(2)] for each \emph{history tree with
canonical identifiers} $(T_i,I_{c,i},l_i)$, by the five steps presented before,
compute its $\sigma$-successor \emph{history tree with canonical identifiers}
$(T_{i+1},I_{c,i+1},l_{i+1})$ for each $\sigma\in \Sigma$.
\item[(3)] (2) is repeated until no new \emph{history trees with canonical identifiers} can
be created.
\end{itemize}
 On this underlying automata structure,
deterministic Rabin transition automata are defined.

\begin{Def} \rm The deterministic Rabin transition automaton equivalent to a given NBW $B$ is $RT=(\Sigma,Q_{RT},Q_{RT0},\delta_{RT},\lambda_{RT})$, where $Q_{RT}$ is the set of \emph{history trees with canonical identifiers} w.r.t $B$; $Q_{RT0}$ the initial \emph{history
trees with canonical identifiers}; $\delta_{RT}$ the transition
relation that is established during the construction of \emph{history
trees with canonical identifiers}; and $\lambda_{RT}=((A_{I1}, R_{I1}),
\ldots, (A_{Ik}, R_{Ik}))$ the Rabin acceptance condition.
\hfill{$\square$}
\end{Def}

In each Rabin pair $(A_{I}, R_{I})$, $I$ ranges over the
canonical identifiers appearing in the \emph{history trees with canonical
identifiers}. $A_{I}$ is the set of transitions through which node
$\tau$ with identifier being $I$ is accepting (annotated by
$\oplus$), while $R_{I}$ is the set of transitions through which
node $\tau$ with identifier being $I$ is unstable (annotated by
$\ominus$).

For an input word $\alpha:\omega\rightarrow \Sigma$, we call the
sequence $\Pi=\anglebra{T_0,I_{c,0},l_0},
\anglebra{T_1,I_{c,1},l_1}, \ldots $ of \emph{history trees with canonical
identifiers} that starts with the initial \emph{history trees with
canonical identifiers} $\anglebra{T_0,I_{c,0},l_0}$ and where, for
every $i\in \omega$, $\anglebra{T_i,I_{c,i},l_i}$ is followed by
$\alpha(i)$-successor $\anglebra{T_{i+1},I_{c,i+1},l_{i+1}}$, the
\emph{history trace with canonical identifiers} of $\alpha$. Let $\Pi$ be
the \emph{history trace with canonical identifiers} of a word
$\alpha$ on the deterministic Rabin transition automaton. $\alpha$
is accepted by the automaton if, and only if, $\mathsf{Inf}(\Pi)\cap
A_{Ii}\not=\emptyset$ and $\mathsf{Inf}(\Pi)\cap R_{Ii}=\emptyset$
for some $(A_{Ii}, R_{Ii})$, $1\leq i\leq k$.

Let $RT$ be the deterministic Rabin transition automaton obtained
from the given non-deterministic B\"{u}chi automaton $B$.

\begin{Thm}\label{thm:transition}
$L(RT)=L(B)$.
\end{Thm}
\begin{proof}
The theorem is proved by appealing to Lemma~\ref{sven1}.\\[1.2ex]
$\Rightarrow$: Given a \emph{history trace with canonical identifiers}
$\widehat{\Pi}$ such that $\mathsf{Inf}(\widehat{\Pi})\cap
A_{Ii}\not=\emptyset$ and $\mathsf{Inf}(\widehat{\Pi})\cap
R_{Ii}=\emptyset$, for some $i\geq 1$. By the construction of
\emph{history trees with canonical identifiers}, there exists some
transition between two \emph{history trees with canonical identifiers} (in
$A_{Ii}$) containing node $\tau$ with $Ii$ being the identifier,
such that $\tau$ is always eventually accepting since
$\mathsf{Inf}(\widehat{\Pi})\cap A_{Ii}\not=\emptyset$. Further
$\tau$ is eventually always stable because
$\mathsf{Inf}(\widehat{\Pi})\cap R_{Ii}=\emptyset$ (in case $\tau$
is unstable in some transition in the loop,
$\mathsf{Inf}(\widehat{\Pi})\cap R_{Ij}\not=\emptyset$).\\[1.1ex]
$\Leftarrow$: Suppose \emph{history trace with canonical identifiers} $\Pi$ contains node $\tau$,
 which is always eventually accepting and
eventually always stable. By the construction of Rabin automata,
there must exists $Ii$ which is the identifier of $\tau$ in the
\emph{history trees with canonical identifiers} such that
$\mathsf{Inf}(\Pi)\cap A_{Ii}\not=\emptyset$ since $\tau$ is always
eventually accepting (thus included in $A_{Ii}$). Further,
$\mathsf{Inf}(\Pi)\cap R_{Ii}=\emptyset$ since $\tau$ is eventually
always stable.
\end{proof}

\begin{Thm}\label{thm:con1}
For a given nondeterministic B\"{u}chi automaton $B$ with $n$
states, we can construct a deterministic Rabin transition automaton
with $o((1.65n)^n)$ states and $2^{\lceil (n-1)/2\rceil}$ accepting pairs that recognizes
the language $L(B)$ of $B$.
\end{Thm}

\begin{proof}
For the state complexity, by the construction of \emph{history trees with
canonical identifiers},  for each node $\tau$ that is possible to occur in
a \emph{history tree} of a B\"{u}chi automaton with $n$ states, a unique
identifier is assigned to the node and keeps unchanged no matter whether it
occurs in a \emph{history tree}. 
 Thus, we cannot find two distinct
\emph{history trees with canonical identifiers} in the determinization
construction of $B$ such that they coincide after the identifiers have
been erased. So, $histf(n)=hist(n)\subset o((1.65n)^n)$.
Here $histf(n)$ is the number of \emph{history
trees with canonical identifiers} for a given
nondeterministic B\"{u}chi automata with $n$ states.

Further, by Lemma \ref{Lem:fullnodes}, for the nodes with height being $h$, totally $2^{h-1}$ different flags are required  in case  $h\leq\lceil (n-1)/2\rceil$; and $2^{n-h-1}$ otherwise.
Thus, $2^{\lceil (n-1)/2\rceil}$ identifiers are required in \emph{$n$-full ordered tree with canonical identifiers}.
As a result, $2^{\lceil (n-1)/2\rceil}$  accepting pairs are
required in the obtained deterministic Rabin transition automata.
\end{proof}

\subsection{From Nondeterministic B\"{u}chi Automata to Ordinary Rabin Automata}
\label{Sec:ordinary}
%
%
%
%
%
%
%
%

By moving the acceptance and unstable
information, i.e.~$\oplus$ and $\ominus$ symbols, from transitions to states (\emph{history trees
with canonical identifiers}), a deterministic Rabin transition automaton can be equivalently transformed as an ordinary deterministic Rabin automaton.
For convenience, we call \emph{history trees with canonical identifiers} decorated with acceptance and unstable information \emph{enriched history trees with canonical identifiers}.

Let $ehistf(n)$  be the number of \emph{enriched history
trees with canonical identifiers} for a given
nondeterministic B\"{u}chi automata with $n$ states. The following lemma can be obtained.

\begin{Lem}\label{Lem:state} For a given nondeterministic B\"{u}chi automaton $B$ with $n$ states,
we have $ehistf(n)\subset o((2.66n)^n)$.
\end{Lem}

\begin{proof}
\emph{Enriched history trees with canonical identifiers} possible occurring in
determinization construction are factually \emph{ordered trees} decorated
with (1) labels, (2) canonical identifiers, and (3) accepting and unstable
information. It can be easily observed that \emph{ordered trees} decorated
with (1) labels and (3) accepting and unstable information, called
\emph{enriched history trees} are exactly the \emph{enriched
history trees} in \cite{Sven09} used by Schewe for defining the
ordinary Rabin acceptance condition. By Theorem \ref{Thm:sven},
$ehist(n)\subset o((2.66n)^n)$. Here $ehist(n)$ denotes the number
of the \emph{enriched history trees} for an NBW with $n$ states. Further,
$ehistf(n)=ehist(n)\subset o((2.66n)^n)$ since additional identifiers
 will not increase the number of trees.
\end{proof}

\begin{Thm}
For a given nondeterministic B\"{u}chi automaton $B$ with $n$
states, we can construct a deterministic Rabin  automaton with
$o((2.66n)^n)$ states and $2^{\lceil (n-1)/2\rceil}$  accepting pairs that recognizes the
language $L(B)$ of $B$. \end{Thm}

\begin{proof}
The state complexity has been obtained in  Lemma~\ref{Lem:state} and
the index complexity can be proved similar to Theorem
\ref{thm:con1}.
\end{proof}

\section{From Nondeterministic B\"{u}chi Automata to Deterministic Parity Automata}\label{sec:parity}
While the determinization construction in the previous section relies on a statical flagging mechanism of nodes in ordered trees, determinization construction in this section is upon a dynamical one where the identifiers of nodes are properly changed throughout the construction process.
The underlying data structure is  \emph{history tree with spinal identifiers} that is actually an ordered tree with spinal identifiers, labels, as well as the information about the minimal\footnote{What smaller means will be discussed later.} accepting and non-accepting nodes.

\begin{Def}[History Trees with Spinal Identifiers]
\rm A \emph{history tree with spinal identifiers}
for a given nondeterministic B\"{u}chi automaton
$B=(\Sigma,Q,Q_0,\delta,\lambda)$ with $n$ states is a tuple
$\anglebra{T,I_s,l,m_g,m_b}$ where

 \begin{itemize}
 \item $\langle T,I_s\rangle$ is an \emph{ordered tree} with spinal identifiers,


     \item $l:T\rightarrow 2^Q\setminus
\{\emptyset\}$ labels each node of $T$ with a non-empty subset of
$Q$, satisfying the following:
\begin{itemize}
\item[{\scshape P1}:] The label of each node is non-empty.
\item[{\scshape P2}:] The labels of different children of a node are
disjoint.
\item[{\scshape P3}:] The label of each node is a proper superset of
the union of the labels of its children.
\end{itemize}

    \item $m_g$, $m_b\in [0..n-1]\times [1..\lfloor\frac{n}{2}\rfloor] $  are used to define the parity acceptance condition. Here $[i..j]$ means the set of integers from $i$ to $k$. $m_g$ is used to memorize the minimal node that is accepting, and $m_b$ the minimal node that is unacceptable in the tree. We notice that for two different nodes $\tau$ and $\tau'$ in a tree, $(h(\tau), flag_s(\tau))<(h(\tau'), flag_s(\tau'))$ if, and only if,  $flag_s(\tau)< flag_s(\tau')$ or $flag_s(\tau)= flag_s(\tau')$ and $h(\tau)<h(\tau')$. \hfill{$\square$}

    \end{itemize}

\end{Def}

%

\begin{Thm}The number of history trees with spinal identifiers for a given nondeterministic B\"{u}chi automaton
$B=(\Sigma,Q,Q_0,\delta,\lambda)$ with $n$ states  is  $o(n^2(0.69n\sqrt{n})^n)$. \label{thm:paritycomplexity}
\end{Thm}
\begin{proof}
 A \emph{History tree with spinal identifiers} is an \emph{ordered tree with spinal identifiers}, labels, as well as $m_g$ and $m_b$.
By Theorem \ref{sven2}, for a given nondeterministic B\"{u}chi automaton
$B=(\Sigma,Q,Q_0,\delta,\lambda)$ with $n$ states, the number of ordered trees with labels (history trees) is $o((1.65n)^n)$. Further, with additional spinal identifiers, by Lemma \ref{Lem:spine}, the number of history trees with spinal identifiers for
$B$ is  $o((1.65n)^n\lfloor\frac{n}{2}\rfloor!)=o((0.69n\sqrt{n})^n)$. For $m_g$ and $m_b$, $n^2$ possible values are required. It follows that the number of history trees with spinal identifiers  for a NBW with $n$ states is $o(n^2(0.69n\sqrt{n})^n)$.
\end{proof}

Fix a nondeterministic B\"{u}chi automaton $B=(\Sigma,Q,Q_0,$
$\delta,\lambda)$. The initial  \emph{history tree with spinal
identifiers}  $\anglebra{T_0,I_{s,0},l_0,m_{g,0},m_{b,0}}$ contains only one node
$\epsilon$ with $I_{s,0}(\epsilon)=1$, label being the set
of initial states $Q_0$, $m_{g,0}=(1,1)$, and $m_{b,0}=(0,1)$.

Given a \emph{history tree with spinal identifiers} $\anglebra{T,I_{s},l,m_g,m_b}$, and an
input letter $\sigma\in \Sigma$ of B\"{u}chi automaton $B$, the
$\sigma$-successor \emph{history tree with spinal identifiers}
$\anglebra{\widehat{T},\widehat{I_s},\widehat{l},\widehat{m_g},\widehat{m_b}}$ of
 $\anglebra{T,I_{s},l,m_g,m_b}$ is constructed following the four steps below.

\paragraph*{\scshape Step 1} We construct the  labeled tree with spinal identifiers $\anglebra{T',I_s',l':T'\rightarrow 2^Q,m_g',m_b'}$ such
that

\begin{itemize}
\item[(1)] $T' = T \cup \{ \tau'=\tau\cdot(deg_{T}(\tau)+1) \mid
 \mbox{ for any } \tau \in T \}$.
That is, $T'$ is formed by taken all nodes from $T$ and creating the
new child node $\tau\cdot(deg_{T}(\tau)+1)$ to any node $\tau \in
T$ if $\delta(l(\tau),\sigma)\cap \lambda\not=\emptyset$;

\item[(2)] the label $l'(\tau)=\delta(l(\tau),\sigma)$ of an old node $\tau\in
T$ is the set $\delta(l(\tau),\sigma)=\bigcup_{q\in
l(\tau)}\delta(q,\sigma)$ of $\sigma$-successors of the states in
the label of $\tau$;

\item[(3)] the label $l'(\tau')=\delta(l(\tau),\sigma)\cap
\lambda$ of a new node $\tau'=\tau\cdot (deg_{T}(\tau)+1)$ is the
set of final $\sigma$-successors of the states in the label of
$\tau$;

\item[(4)] Set the spinal flag of a new created node $\tau$  to be $min(\{flag_s(\tau')\mid \tau'$ is a leaf node with no younger siblings in $T$, and  $ C(\tau)\supset C(\tau')\}\cup \{$the minimal integer greater than the largest spinal flag used in the current tree$\})$.


\item[(5)] Set $m_g'$ and $m_b'$ to be undefined.

%

\end{itemize}

\paragraph*{\scshape Step 2} In this step, {\scshape P2} is re-established. We construct the
tree $\anglebra{T',I_s',l'':T'\rightarrow 2^Q,m_g',m_b'}$, where $l''$ is
derived from $l'$ by removing each state {$q$} in the label of a
node $\tau$ and all its descendants if {$q$} appears in the label
$l'(\tau)$ of an older sibling.

\paragraph*{\scshape Step 3} {\scshape P3} is re-established first in
the this step. We construct the tree $\anglebra{T'', I_s'', l'':
T''\rightarrow 2^Q, m_g'',m_b''}$ by removing all descendants of nodes $\tau$ if
the label of $\tau$ is neither empty nor a proper superset of the union of the
labels of the children of $\tau$.
The nodes whose descendant have
been deleted because of this rule are called accepting (\emph{good}) nodes.
Next, {\scshape P1} is re-established
 by  removing all nodes $\tau$ with an empty
label $l''(\tau)=\emptyset$. A removed node in this step is an unacceptable (\emph{bad}) node.
 $m_g''=min(\{(h(\tau), flag_{s}(\tau))\mid v\mbox{ is accepting in this step}\}\cup\{(1,\lceil\frac{n}{2}\rceil)\})$, and $m_b''=min(\{(h(\tau), flag_s(\tau))\mid v\mbox{ is removed in this step} \}\cup\{(1,\lceil\frac{n}{2}\rceil)\})$.




\paragraph*{\scshape Step 4} The identifier, i.e. height and spinal flag, of each node is adjusted in this step.
 First, the underling tree resulting from {\scshape Step 3} satisfies properties {\scshape P1-P3}, but it is possible that the tree is not order-closed. We construct $\sigma$-successor
$\anglebra{\widehat{T},\widehat{I_s},\widehat{l}: \widehat{T}\rightarrow
2^Q\setminus \{\emptyset\}, \widehat{m_g},\widehat{g_b}}$ of $\anglebra{T,I_s,l,m_g,m_b}$
  by ``compressing" $T''$ to an order-closed tree, using the
compression function $comp:T'' \rightarrow \omega^*$ that maps the
empty word $\epsilon$ to $\epsilon$, and $\tau\cdot i$ to
$comp(\tau)\cdot j$, where $j=|\{k<i\mid \tau\cdot k\in T''\}|+1$
is the number of older siblings of $\tau\cdot i$ plus $1$. As a result, for each node $\tau$, the height $h(\tau)$ is updated accordingly.  For each spine $sp$ with spinal flag being $i\geq 1$, if there are nodes in $sp$ whose hight is updated, the spinal flag of the largest node in $sp$ is reset to the minimal integer greater than the largest spinal flag used in the current tree.
Then, for each spine $sp$ with spinal flag being $i\geq 1$, if there are nodes in $sp$ removed in Step 3, the spinal flag of the largest node in $sp$ is reset to the minimal integer greater than the largest spinal flag used in the current tree.
Subsequently, we make the spinal flags  to be consecutive integers starting from $1$ by subtracting the number of vacant integers smaller than the spinal flag of largest node in a spine. Finally, the spinal flags of other nodes are updated according to the definition of ordered tree with spinal flags.



\begin{Def} \rm The deterministic Parity automaton equivalent to a given NBW $B$ is $P=(\Sigma,Q_{P},Q_{P0},\delta_{P},\lambda_{P})$, where $Q_{P}$ is the set of \emph{history trees with spinal identifiers} w.r.t $B$; $Q_{P0}$ the initial \emph{history
trees with spinal identifiers}; $\delta_{P}$ the transition
relation that is established during the construction of \emph{history
trees with spinal identifiers}; and $\lambda_{P}=(\lambda_1,\cdots,\lambda_n)$ the parity acceptance condition defined below:
$$
\begin{array}{rlll}

\lambda_{2i-1}=\{\anglebra{T,I_{s},l,m_g,m_b}\in Q_{P}\mid m_g.flag_s=i \mbox{ and } m_g\geq m_b\}\\
\lambda_{2i}=\{\anglebra{T,I_{s},l,m_g,m_b}\in Q_{P}\mid m_g.flag_s=i \mbox{ and } m_g< m_b\}\\
\end{array}
$$
where $1\leq i\leq \lfloor\frac{n}{2}\rfloor$.
\hfill{$\square$}
\end{Def}


\begin{Thm}\label{thm:p}
$L(P)=L(B)$.
\end{Thm}
\begin{proof}
The theorem is also proved by appealing to Lemma~\ref{sven1}.\\[1.2ex]
$\Rightarrow$: Given a \emph{history trace with spinal identifiers}
$\widehat{\Pi}$ such that for some even number $2i$, $i\geq 1$, $\mathsf{Inf}(\widehat{\Pi})\cap
\lambda_{2i}\not=\emptyset$, and, for all $j$ where $1\leq j<2i$, $\mathsf{Inf}(\widehat{\Pi})\cap
\lambda_{j}=\emptyset$.  
By the construction of
\emph{history trees with spinal identifiers}, there exists some
 \emph{history trees with spinal identifiers} $T$ (in
$\lambda_{2i}$) containing node $\tau$ with $(h(\tau),i)$ being the identifier,
such that $\tau$ is always eventually accepting since
$T\in \mathsf{Inf}(\widehat{\Pi})\cap \lambda_{2i}$. 
Further
$\tau$ is eventually always stable since in case it is unstable, its height will be updated, and thus the spinal flag will be reset to a larger number by the construction of history trees with spinal identifiers.  \\[1.2ex]
$\Leftarrow$: Suppose \emph{history trace with spinal identifiers} $\Pi$ contains nodes 
 which is always eventually accepting and
eventually always stable. Let $\tau$ be the smallest one of them. 
By the construction of Parity automata,
there must exists a history tree with spinal identifiers $T$ containing node $\tau$ with identifier being $(h(\tau), i)$ such that
$\mathsf{Inf}(\Pi)\cap \lambda_{2i}\not=\emptyset$ since $\tau$ is always
eventually accepting and stable.
Further,
$\mathsf{Inf}(\Pi)\cap \lambda_{j}=\emptyset$, $1\leq j<2i$, since $\tau$  is the smallest node that is always eventually accepting and eventually always stable. 
\end{proof}

%
%

By Theorem \ref{thm:paritycomplexity} and \ref{thm:p}, the following Corollary can be obtained.

\begin{Cor}\label{thm:con1}
For a given nondeterministic B\"{u}chi automaton $B$ with $n$
states, we can construct a deterministic Parity transition automaton
with $o(n^2(0.69n\sqrt{n})^n)$ states and $n$ accepting pairs that recognizes
the language $L(B)$ of $B$.
\end{Cor}

\section{Conclusion} \label{Sec:con}
In this paper,  we present a transformation from B\"{u}chi automata
of size $n$ into deterministic Rabin transition automata with
$o((1.65n)^n)$ states and $2^{\lceil (n-1)/2\rceil}$ Rabin pairs.  We also present a
construction of deterministic Parity automata from NBWs with
$o(n^2(0.69n\sqrt{n})^n)$ states and $n$ priorities. Thus, the
current best construction with the same (and optimal) state
complexity, but $2^{n-1}$ Rabin pairs, and $O(2(n!)^2)$ state complexity and $2n$ Parity priorities are improved.

\end{document}